\documentclass[12pt]{article}
\usepackage{fullpage,times}

\usepackage[utf8]{inputenc}
\usepackage[T1]{fontenc}

\usepackage{amsmath}
\usepackage{amssymb}
\usepackage{algorithmic}
\usepackage{amstext}
\usepackage{amsfonts}
\usepackage{amsthm}
\usepackage{mathtools}
\usepackage{scrextend}    
\usepackage{array}     
\usepackage{multicol}    
\usepackage{mathrsfs}    
\usepackage{color}
\usepackage{hyperref}
\usepackage{cleveref}
\usepackage{bm}
\usepackage{enumitem}
\usepackage{upgreek}
\usepackage{mleftright}
\usepackage[mathscr]{euscript}
\usepackage{thm-restate}
\usepackage[linesnumbered,ruled,vlined]{algorithm2e}

\newtheorem{theorem}{Theorem}[section]    % Specify Theorem
\newtheorem{definition}{Definition}[section] % Specify Definition
\newtheorem{corollary}[theorem]{Corollary}    % Specify Corollary
\newtheorem{lemma}[theorem]{Lemma}    % Specify Lemma

\renewcommand{\qed}{\hfill{$\rule{6pt}{6pt}$}} %Box at end of proof
\renewenvironment{proof}{\noindent{\bf Proof:}}{\qed\\}
\newenvironment{proofof}[1]{\noindent{\bf Proof of #1:}}{\qed\\}
\numberwithin{equation}{section}

\newcommand{\complex}{{\mathbb C}}

\newcommand{\norm}[1]{\left\| #1 \right\|}

\newcommand{\Tr}{\mathrm{Tr}}

\newcommand{\Pos}{{\mathsf{Pos}}}

\newcommand{\ket}[1]{| #1 \rangle}

\newcommand{\ketbra}[2]{| #1 \rangle\!\langle #2 |}

\newcommand{\eqdef}{\coloneqq}

\newcommand{\suppress}[1]{}
\newcommand{\ith}{^{\text{th}}}
\newcommand{\ii}{\mathrm{i}}

\newcommand{\cA}{\mathcal{A}}
\newcommand{\cB}{\mathcal{B}}

\newcommand{\cD}{\mathcal {D}}

\title{\textbf{Dimension Independent and Computationally Efficient Shadow Tomography}}

\author{
Pulkit Sinha~\thanks{School of Computer Science,
and Institute for Quantum Computing, University
of Waterloo, 200 University Ave.\ W., Waterloo, ON,
N2L~3G1, Canada.
Email: \texttt{psinha@uwaterloo.ca}~.
} \\
Institute for Quantum Computing,\\
University of Waterloo}

\date{3 November 2024}

\begin{document}

\maketitle

\begin{abstract}
We describe a new shadow tomography algorithm that uses $n=\Theta(\sqrt{m}\log m/\epsilon^2)$ samples, for $m$ measurements and additive error $\epsilon$, which is independent of the dimension of the quantum state being learned. This stands in contrast to all previously known algorithms that improve upon the naive approach. The sample complexity also has optimal dependence on $\epsilon$.  Additionally, this algorithm is efficient in various aspects, including quantum memory usage (possibly even $O(1)$), gate complexity, classical computation, and robustness to qubit measurement noise. It can also be implemented as a read-once quantum circuit with low quantum memory usage, i.e., it will hold only one copy of $\rho$ in memory, and discard it before asking for a new one, with the additional memory needed being $O(m\log n)$. Our approach builds on the idea of using noisy measurements, but instead of focusing on gentleness in trace distance, we focus on the \emph{gentleness in shadows}, i.e., we show that the noisy measurements do not significantly perturb the expected values.
\end{abstract}

\section{Introduction}
\label{sec-introduction}
Consider the following statistical estimation problem: given $m$ events $E_1,E_2\dots E_m$, and i.i.d samples from an unknown probability distribution $\cD$ defined on the underlying sample space, estimate each $\Pr[E_i]$ up to an additive error $\epsilon$. This task is well studied, and it is known that it can be solved with just $\Theta(\frac 1 {\epsilon^2}\log \left(\frac m \delta\right))$ samples of $\cD$ (which we refer to as the \emph{sample complexity}), where $\delta$ is the allowed probability of failure. The algorithm uses the empirical mean estimator for each event $E_i$: for each sample $D_j\sim \cD$, record whether $E_i$ occurred or not, and then output the frequency divided by $n$ as the estimate for $\Pr[E_i]$. The distribution of this output is going to be $\mathrm{Binomial}(n,\Pr[E_i])/n$, which is known to concentrate around its mean. The dependence on $m$ in the expression is a by-product of the dependence on $\delta$: for a single event $E$, $\Theta\left(\frac 1 {\epsilon^2} \log \frac 1 \delta \right)$ samples are sufficient. So for $m$ events, one simply uses the algorithm designed for probability of failure $\frac \delta m$, for each event $E_i$, which due to the union bound, has an overall probability of failure at most $\delta$. The key point in this reduction is that the algorithm can independently estimate each probability, and only the probability of error needs to be tracked.

In the quantum case, probability distributions are generalized by quantum states $\rho\in D(\mathbb C^d)$, and events are generalized as POVM elements $M_i\in \mathrm{Pos}(\mathbb C^d)$, $M_i\leq I$. The corresponding problem in the quantum case is as follows: 
\begin{center}
    \textit{Given $m$ POVM elements $M_1,M_2\dots M_m$, and samples of an unknown quantum state $\rho$, estimate each $\Tr[\rho M_i]=\mathbb E_\rho[M_i]$ upto an additive error of $\epsilon$.}
\end{center} 
This is called the \emph{shadow tomography} problem.  Note that any algorithm solving this problem can also estimate mean values of observables, i.e., given Hermitian operators $O_1,O_2\dots O_n\in L(\mathbb C^d)$, we can reduce the problem of estimating each $\Tr[O_i\rho]$ to shadow tomography by appropriately scaling and shifting each $O_i$ to between $0$ and $I$.

This problem was introduced by Aaronson \cite{aaronson2018} (and was termed so by Steve Flammia) because of its relation to the more general problem of \emph{quantum state tomography}, where the objective is to instead learn the state $\rho$ up to an error of $\epsilon$ in trace distance. One can rephrase this problem as trying to find an estimate $\rho'$ for $\rho$ such that for any POVM element $M$, $|\mathbb E_\rho[M]-\mathbb E_{\rho'}[M]|\leq \epsilon$, which means that we effectively need the estimate for each possible $\mathbb E_\rho[M]$ to be within $\epsilon$. Aaronson referred to these $\mathbb E_{\rho}[M]$ as \emph{shadows}. In many scenarios where one considers learning a quantum state, they might only be interested in learning the values of a few specific shadows: any property of $\rho$ that can be expressed as a probability of a quantum circuit's output being $1$ (on the input being $\rho$) can be expressed as a shadow of $\rho$.

For $m=1$, the classical algorithm generalizes directly: $n=\Theta(\frac 1 {\epsilon^2} \log \frac 1 \delta)$ samples suffice, as one can perform the POVM $\{M,I-M\}$ on each sample, reducing the problem to the classical case. However, this approach does not work for $m\geq 2$, as there is no notion in general of ``simultaneously" measuring different POVMs together, although it is possible in certain special cases, for example, when the $M_i$'s commute. To directly reuse the classical algorithm, we would, in general, need to use new copies for each estimation, bringing the sample complexity to $\Theta\left(\frac m {\epsilon ^2} \log \frac m \delta \right)$. When $m$ is large, this is very inefficient compared to the classical algorithm. This raises a natural question: can we do better than just simply asking for new copies each time?

In many cases, the answer is yes. Numerous algorithms have been proposed for this problem, each with its own strengths and weaknesses. We discuss a few of them here.

Aaronson gave the first algorithm \cite{aaronson2018}, which achieved sample complexity $\Tilde O\left(\frac{\log^4 m \log d}{\epsilon^4}\right)$, which was based on his previous work on postselected learning of quantum states, and the Quantum OR bound. Then, Huang, Kueng, and Preskill \cite{Huang2020} published their now very well-known, classical shadows algorithm. Technically, it solves a different problem, but it can nonetheless be phrased as an algorithm solving shadow tomography. Given copies of $\rho$, it outputs a classical transcript that can be used to estimate the expectation of \emph{any} set of $m$ observables/POVM elements with bounded \emph{shadow norm}, say by $M$, with sample complexity $O\left(\frac {M^2}{\epsilon^2} \log \frac m \delta\right) $. Here, the shadow norm depends on the specific implementation of the algorithm, and examples of operators that have low shadow norm (with respect to some corresponding implementation) are $k$-local observables and low-rank observables. For arbitrary POVM elements, this shadow norm can be as large as $\Omega(\sqrt d)$, so the sample complexity of this algorithm for arbitrary POVMs has sample complexity $\Theta\left(\frac d {\epsilon^2} \right)$, scaling exponentially large in the number of qubits. Nonetheless, this is still much better than the $\Omega\left(\frac {d^2}{\epsilon^2}\right)$ copies required for full state tomography. Additionally, it uses only non-adaptive single-copy measurements. 

Next, Aaronson and Rothblum \cite{aaronsonrothblum2019} provided an algorithm that achieved sample complexity $\Tilde{O}\left(\frac{\log ^2 m\log^2d}{\epsilon^4}\right)$, improving upon Aaronson's original algorithm in its $m$ dependence. They named this as the \emph{Quantum Private Multiplicative Weights} algorithm. Their key idea was to solve shadow tomography using differentially private learning: their algorithm measured the sample mean with a Laplace noise, showing that such measurements are "gentle", through their connection to differential privacy. Finally, they used the Matrix Multiplicative Weights algorithm \cite{OLQS} (M.M.W.) as a differentially private learner.

Later, Badescu and O'Donnell \cite{badescu} refined this approach, reducing the sample complexity to $\Tilde{O}\left(\frac{\log^2m\log d}{\epsilon^4}\right)$. The high-level descriptions of both algorithms are quite similar, the main difference being that Badescu and O'Donnell instead related the analysis to ideas in adaptive learning. They separated the algorithm into two components: Threshold Search, which uses the noisy measurements to figure out wrong estimates, and the M.M.W. algorithm to maintain a guess state which is updated whenever the threshold search detects a mistake. Unlike Aaronson and Rothblum's algorithm, the threshold search routine asks for fresh copies of $\rho$ each time it detects a mistake, and uses exponential noise instead of Laplace. 

Later, Bostanci and Bene Watts \cite{Bostanci} showed an alternate algorithm that solved Threshold search by using their improved Quantum Event Finding algorithm, which led to the same sample complexity as Badescu and O'Donnell's.

Finally, we also have the relatively new algorithm by Chen, Li and Liu \cite{chen2024}, which in the regime $\epsilon^{-1}\geq \mathrm{poly}(d)$ has the sample complexity $\Theta\left(\frac {\log m}{\epsilon ^2}\right)$, i.e, it achieves the best possible sample complexity in the extremely high precision regime. This algorithm, like the classical shadows algorithm, does not need the POVM elements beforehand: it outputs a transcript that can be used to estimate the expected values of any set of $m$ POVM elements. Their algorithm was based on their previous work on a representation-theoretic algorithm for full-state tomography.

All the discussed algorithms have some form of inherent dimension dependence, which can somewhat be boiled down to the fact that all of them treat shadow tomography as a form of weak state tomography, i.e, they are inherently trying to perform full state tomography, but weakly, as they effectively maintain a classical representation of the guess of the quantum state. Also note that because of this, these algorithms are computationally inefficient, as the estimations require $\mathrm{poly}(d)$ time classical operations. Although Brandao et. al. \cite{brandao_et_al} did show that the first algorithm by Aaronson can indeed be made efficient for low-rank observables, it is still inefficient in the general case. So, can we overcome this inherent dimension dependence?

To some extent, yes. In this work, we present a shadow tomography algorithm that has sample complexity sublinear, but still polynomial, in $m$, but has no dependence on $d$. It is also computationally efficient in many aspects:

\begin{theorem}\label{main_thm}
    There is a shadow tomography algorithm that, given $m$ POVM elements $M_1,M_2\dots M_m$, and $\epsilon,\delta>0$, and copies of unknown state $\rho$, outputs each $\mathbb E_\rho[M_i]$ with additive error at most $\epsilon$, that has sample complexity at most $n=\Theta\left(\frac {\sqrt m}{\epsilon^2} \log \frac 1 \delta\right)$. Furthermore, depending on the implementation, it can have the following additional properties:
    \begin{itemize}
        \item Given Quantum Circuits of size $S_i$ implementing each $M_i$, the quantum circuit for the algorithm has size $\Theta\left(\frac {nm} {\epsilon^2}\log \frac 1 \delta +n \sum S_i\right)$ .
        \item It can be implemented as a read-once quantum circuit, i.e, it can only access one copy of $\rho$ at a time, with additional memory $\Theta\left(\frac m {\epsilon^2}\log \frac 1 {\epsilon} \right)$.
        \item It can also be implemented as a quantum circuit with only $O(1)$ ancilla qubits (apart from those needed to implement the $M_i$), assuming access to mid-circuit measurements, but with an overhead of $\Theta\left( \frac n {\epsilon^2}\log \left(\frac 1 \delta\right) \sum S_i\right)$ gates. This overhead reduces to $\Theta( \frac {n+m} {\epsilon^2}\log \frac 1 \delta \log n)$ if we have access to additional $O(\log n)$ ancilla qubits.
        \item There is no non-trivial classical post-processing needed.
        \item All the implementations are robust against qubit measurement noise, i.e., measurement noise only impacts the accuracy of the estimates.
    \end{itemize}
\end{theorem}

In Section \ref{sec_low_mem}, we discuss an alternative algorithm based on similar ideas that achieves the same sample complexity and has slightly better scaling in gate complexity and memory usage, with the main drawback of not being robust to qubit measurement noise. Most notably, that algorithm allows the memory usage in the read-once circuit implementation to be reduced to $O(m\log n)$.

The sample complexity obtained is currently the best known in the regime $m=o\,(\log^2 d)$.
The classical post-processing referenced above refers simply to averaging, square roots, and $\arcsin$. The key idea is, as in \cite{aaronsonrothblum2019, badescu}, to use noisy measurements, however, instead of bounding the trace distance between the pre and post-measurement states, we try to see how the expected values change. Our algorithm can be viewed as a sequence of independent estimations of the different $M_i$, each step reusing all the copies of $\rho$. Upon measuring the first estimate, there is a clear collapse in the state, leading the resultant to be no longer a product state. However, if we don't condition on any one particular value of the estimate, then we find that the post-measurement state can be written as a mixture of i.i.d copies of some quantum state, which with a high probability have expected values very close to the original state $\rho$. This leads to a union-bound style argument similar to the classical case.

In our exact implementation, for each estimation, we create a noisy encoding of the estimate in ancilla qubits and measure them on a computational basis to finally obtain the estimates. On the other hand, the unitary operator that creates the encoding commutes with Hadamard basis measurements, so if we, for the sake of analysis, measure in the Hadamard basis instead, then we find that this procedure has induced a more general form of phase kickback onto the copies of $\rho$, and this phase kickback itself has a product structure, which turns out, with high probability, does not change the expected values for other measurements much. One can think of this as a form of \emph{gentleness in shadows}, where instead of small deviation with respect to trace distance, it is the shadows that are left relatively unperturbed.

Notice that we are only making claims about the expected values. This is because, with our measurement, the post-measurement state is indeed very far from the initial state, however for the purposes of performing shadow tomography this change is immaterial as long as the expected values, or the shadows, do not change much. 

Due to its various properties, we envision our algorithm as more practical than the other discovered algorithms in certain scenarios. For example, consider a quantum algorithm that involves the creation of a specific quantum state $\rho$, which is then repeated multiple times to estimate some $m$ predetermined properties. In this, say the creation of the state is the main bottleneck. Then, our algorithm leads to a quadratic speed up as instead of repeating the creation process $\Theta(m\log m/\epsilon^2)$ many times, we only need to do it $\Theta(\sqrt{m}\log m/\epsilon^2)$ times, and the additional gate complexity is still overshadowed by the complexity of the creation process. However, if we happen to use any of the other known algorithms, even though we invoke the creation process a much smaller number of times, the additional gate complexity and classical computation needed might completely negate any of the speedups we get.

Usually, for such learning problems, the best one can hope for is an efficient single-copy measurement procedure. However, since it was proven in \cite{chen_lb} that one cannot perform better than the naive algorithm with respect to sample complexity without access to quantum memory, single-copy measurements will not provide any speedup at all. The next best thing one can then hope for is a read-once circuit, which only has access to 1 copy of $\rho$ at a time (which is then immediately discarded before the next copy is received). For such a circuit, the additional quantum memory used becomes another important parameter of computational hardness, which, in our case, is again independent of the dimension. If, instead, the algorithm is allowed to revisit the same copy, then just $O(\log n)$ extra memory suffices for our algorithm without any additional gate overhead dependent on  $S_i$'s. If even such an overhead is alright, then our algorithm can be implemented with just $O(1)$ ancilla qubits.

We end with an example of a restriction on the $m$ POVMs that leads to better sample complexity. If all the spectral norms of the commutators are bounded by some small $C_{max}$, then the sample complexity reduces down to $n=O\left(  (C_{max}\sqrt m+1)\frac{1}{\epsilon^2} \log \frac {m}\delta+ \sqrt{C_{max}}\frac{\sqrt{m\log (m/\delta)}+\log(m/\delta)}{\epsilon^{1.5}} \right)$. We sketch out the proof of this by outlining the inequalities that can be tightened in the main proof.
\paragraph{Acknowledgements:}
This research was supported in part by NSERC Canada and by a Mike and Ophelia Lazaridis Fellowship. We thank Ashwin Nayak and Ryan O'Donnell for helpful discussions. 
\paragraph{Other related works:} There has also been some work on restricted forms of shadow tomography. We include a non-exhaustive list for reference. See \cite{king2024triplyefficientshadowtomography,chen2024optimaltradeoffsestimatingpauli,Grier2024sampleoptimal,huang2024}. In particular, \cite{huang2024} achieves very similar bounds for adaptively chosen local and Pauli observables. For a recent survey on Quantum Learning, see \cite{Anshu2024}.
\section{Preliminaries}
\label{sec-preliminaries}

\paragraph{Quantum Information notation.}

For a thorough introduction to the basics of quantum information, we refer the reader to the book by Watrous~\cite{W18-TQI}. We briefly review the notation that we use in this article. 

We consider only finite-dimensional Hilbert spaces in this work and denote them either by capital script letters like~$\cA$ and $\cB$, or directly as~$\complex^m$ for a positive integer~$m$.
A \emph{register\/} is a physical quantum system, and we denote it by a capital letter, like~$A$ or~$B$. A quantum register~$A$ is associated with a Hilbert space~$\cA$, and the state of the register is specified by a unit-trace positive semi-definite operator on~$\cA$. The state is called a \emph{quantum state\/}, or also as a \emph{density operator\/}. We denote quantum states by lowercase Greek letters, like~$\rho$ and~$\sigma$.  We use notation such as~$\rho^A$ to indicate that register~$A$ is in state~$\rho$ and may omit the superscript when the register is clear from the context. We use \emph{ket\/} and \emph{bra\/} notation to denote unit vectors and their adjoints in a Hilbert space, respectively. For two linear operators $A,B$ defined on the same Hilbert space $[A,B]\eqdef AB-BA$ denotes the \emph{commutator} of the two operators. It is to be noted that if $A$ and $B$ are Hermitian, then $[A,B]$ is skew-Hermitian, i.e., $\ii[A,B]$ is a Hermitian operator.

For a linear operator~$M$ on some Hilbert space, $\norm{M}$ denotes the operator norm with respect to Euclidean distance: 
        \[\norm{M}= \max_{\ket{\psi}\ne 0} \frac{\norm{M\ket \psi}_2}{\norm{\ket \psi}_2}.\]
This coincides with the definition of Schatten-$\infty$ norm, which is the largest singular value of $M$. For normal operators, this is the same as the largest magnitude of any eigenvalue of M.

The states evolve by the action of unitary operators $U\in L(\mathbb C^d)$, where the action is $\rho\mapsto U \rho U^{\dagger}$. The Pauli $X$ gate refers to the unitary operator with the matrix
    \[X\eqdef\begin{pmatrix}
        0&1\\
        1&0\\
    \end{pmatrix}.\]

    Its eigenvectors are $\ket +\eqdef \frac {1}{\sqrt 2} (\ket 0 +\ket 1)$ and $\ket -\eqdef \frac {1}{\sqrt 2} (\ket 0 -\ket 1) $, with eigenvalues $+1$ and $-1$ respectively. Corresponding to this, the gate $R_X$ parametrized by $\theta$ is defined as 

    \[R_X(\theta)\eqdef \exp\left(-\frac \ii 2 X\right)= \begin{pmatrix}
        \cos\left(\frac \theta 2\right)& -\ii\sin\left(\frac \theta 2\right)\\
        -\ii\sin\left(\frac \theta 2\right)&\cos\left(\frac \theta 2\right)
    \end{pmatrix}.\]

    For any unitary gate $U$, we can define the corresponding controlled version as $CU\eqdef \ketbra{0}{0}\otimes I+\ketbra 11 \otimes U$

    We now define two notions of measurement. A \emph{Positive Operator Valued Measurement} (POVM) refers to set of positive matrices $\mathcal M=\{P_1,P_2\dots P_m\}\subset \Pos(\mathbb C^d)$ such that 
    \[\sum_{i=1}^{m} P_i=I.\]

    Upon measuring $\rho$ with POVM $\mathcal M$, the $i\ith $ outcome is obtained with probability $\Tr[P_i\rho]$. In the canonical implementation of this POVM, the post-measurement state is $\frac{\sqrt{P_i}\rho{\sqrt P_i}}{\Tr[P_i\rho]}$\\

    The other notion is of the measurement of an observable: given a Hermitian operator $M$, it can be measured as an \emph{observable} using the POVM $\mathcal M=\{P_\lambda:\lambda\in \mathrm{spectrum}(M)\}$, where $P_\lambda$ refers to the projector onto the $\lambda$-eigenspace of $M$. The outcome of this measurement is treated to be $\lambda$ itself, so this way the expected value of observed $\lambda$ upon measurement turns out to be exactly $\Tr[M\rho]$. Notice that for any POVM element $P$ (i.e., $0\leq P\leq I$), $\Tr[P\rho]$ is also the probability of observing the corresponding outcome in a POVM containing $P$. Keeping this in mind, we define $\mathbb E_{\rho}[M]\rho\eqdef \Tr[M\rho]$ for all Hermitian operators, noting that it also refers to the probability of observing the corresponding outcome if it is also a POVM element.
\\

\paragraph{Distributions and Tail bounds:} These definitions and results are discussed in much more detail in chapter 2 of the book by Vershynin \cite{HDP}.
\begin{lemma}[Markov's Inequality,\cite{HDP}]
For distribution $X$ over the nonnegative real numbers, \emph{Markov's inequality} says that for any positive real $a$ we have
\[\Pr[X \geq a]\leq \frac 1 a \cdot \mathbb E[X]\]
\end{lemma}

\begin{definition}
    The \emph{moment generating function} (M.G.F.) of $X$ is the the following function (for $a\in \mathbb R$):
\[f(a)=\mathbb E[ \exp(aX)]\]
\end{definition} 

\begin{lemma}[Hoeffding's inequality\cite{Hoeffding1963}]\label{lem_hoeffding} 
    For $X_1,X_2\dots X_n$ being $n$ independent random variables, each supported on $[a_i,b_i]$ respectively, then for the sum $S=\sum_{i=1}^n X_i$, we have for any $t\geq 0$:
    \[\Pr[|S-E[S]|\geq t]\leq 2 \exp\left(-\frac {2t^2}{\sum_{i=1}^n{(b_i-a_i)}^2}\right).\]
\end{lemma}
\begin{corollary}
    For any $\epsilon,\delta>0$, with $n=\Theta\left(\frac {\log \frac 1 \delta}{\epsilon^2}\right)$ i.i.d. copies $X_1,X_2\dots X_n$ of $X$ distributed on $[0,1]$, we have that 
    \[\Pr\left[\left|\frac 1 n\sum_{i=1}^n X_i-E[X]\right|\geq \epsilon\right]\leq \delta.\]
\end{corollary}
\begin{proof}
    Via a direct application of Lemma \ref{lem_hoeffding}, we get
    \[\Pr\left[\left|\sum_{i=1}^n X_i-nE[X]\right|\geq t\right]\leq 2\exp\left(-\frac {2t^2}{n}\right)\]
    Taking $t=n\epsilon$, and dividing by $n$ inside $\Pr$, we get
    \[\Pr\left[\left|\frac 1 n\sum_{i=1}^n X_i-E[X]\right|\geq \epsilon\right]\leq 2\exp\left(-2\epsilon^2n\right).\]

    The R.H.S can be made at most $\delta$ with $n=\Theta\left(\frac {\log \frac 1 \delta}{\epsilon^2}\right)$.
\end{proof}
There is also a generalization of Hoeffding's inequality, called Azuma's inequality
\begin{lemma}[Azuma's inequality, rephrased\cite{HDP}]\label{lem_azuma}
Given a sequence of random variables $X_1,X_2\dots X_n$, such that $|X_i|\leq c_i$, and $\mathbb E[X_i|X_{i-1},X_{i-2}\dots X_1]=0$, then we have

\[\Pr\left[\left|\sum_{i=1}^n X_i\right| \geq \epsilon\right]\leq 2 \exp\left(-\frac{\epsilon^2}{2\sum_{i=1}^n c_i^2}\right)\]
\end{lemma}
\begin{definition}
    The \emph{Sub-Gaussian norm} of a distribution $X$ on $\mathbb R$ is defined as follows:

\[\norm{X}_{\psi_2}=\inf \{t:\mathbb E[X^2{t^{-2}}]\leq 2\}.\]
\end{definition} 

Any distribution with a finite Sub-Gaussian norm is said to be Sub-Gaussian. For concentration inequalities, the Sub-Gaussian norm being small implies tighter concentrations.

\begin{lemma}[Sums of independent Sub-Gaussians \cite{HDP}]\label{subg_sum}
    Given $n$ independent Sub-Gaussian random variables $X_1\dots X_n$, with mean 0, then we have that for some absolute constant $C$(independent of $n$),

    \[\norm{\sum_{i=1}^n X_i}^2_{\psi_2}\leq C\sum_{i=1}^n\norm{X_i}_{\psi_2}^2.\]
\end{lemma}

 The other technical lemmas we will use are:
 \begin{lemma}\label{lipschitz}
        $\arcsin(\sqrt x)$ is $ \frac 2 {\sqrt 3}$-Lipschitz on $(\frac 1 4, \frac 3 4)$.
    \end{lemma}
    \begin{proof}
        The derivative of $\arcsin(\sqrt x)$ is $\frac 1 {2\sqrt {x(1-x)}} $, which is at most $\frac 2 {\sqrt 3}$ on the given interval.
    \end{proof}

    \begin{lemma}\label{lem_conjugate_bound}
    For hermitian operators $A,B$ with $\norm{A}\leq 1$, 
        \[\norm{e^{iA}Be^{-iA}-B-i [A,B]}\leq \norm{[A[A,B]} \cdot\frac{e^2-3}4 \]
    \end{lemma}
    \begin{proof}
        This is proved using a result proved by Campbell, called the Campbell identity \cite{Hall_2016}:
        \[e^{A}Be^{-A}=B+[A,B]+\frac 1 2 [A[A,B]]+ \frac 1 {3!} [A[A[A,B]]]+\dots \frac 1 {j!}[A[A\dots A,[A,B]]\dots]]+\dots\]

        where the $j\ith$ has $j$ composed commutators. Using this, and $\norm{[X,Y]}\leq 2\norm{X}\norm{Y}$, we get that
    \begin{align*}
        \norm{e^{iA}Be^{-iA}-B-i [A,B]}&=\norm { \frac 1 {2!} [iA[iA,B]]\frac 1 {3!} [iA[iA[iA,B]]]+\dots \frac 1 {i!}[iA[iA\dots iA,[iA,B]]\dots]]}    \\
        &\leq \sum_{j=2}^\infty \frac 1 {j!} \norm{[iA[iA\dots iA,B]\dots]}\\
        &\leq  \norm{[A[A,B]]}\sum_{j=2}^\infty \frac {2^{j-2}} {j!}\\
        &=\norm{[A[A,B]]}\cdot\frac{e^2-2-1}4  
        \end{align*}
        
    \end{proof}
    \begin{lemma}\label{lem_cos_bound}
        For $x\in [-\pi/2,\pi/ 2]$

        \[\cos(x)\leq 1-\frac {x^2} 2 +\frac {x^2}{24}\leq 1-\frac{x^4}{4}\]
    \end{lemma}
    \begin{proof}
        WLOG suppose $x\in [0,\pi/2]$, as $\cos(x)$ is an even function.

        Then, the first inequality follows by differentiating 5  times. The L.H.S becomes $-\sin(x)\leq 0$ on $[0,\pi /2]$, while the RHS is 0. The first 4 derivatives at 0 match up.

        For the second inequality, we rewrite the middle expression as $1-\frac {x^2}{2} \cdot \left(1-\frac {x^2}{12}\right)$, and we have that $\frac{x^2}{12}\leq \frac {\pi^2}{48}\leq \frac {10}{48}$, so we get that 
        \[ 1-\frac {x^2}{2} \cdot \left(1-\frac {x^2}{12}\right)\leq 1-\frac {x^2}{2}\cdot \frac{38}{48}\leq 1-\frac {x^2}{4}\]
    \end{proof}
    \begin{lemma}\label{lem_e-x_bound}
    For $x\geq 0$, we have
    \[e^{-x}\leq 1-x+\frac {x^2}{2}\]
        
    \end{lemma}
    \begin{proof}
        Take 3 derivatives on both sides. The LHS becomes $-e^{-x}\leq 0$, and the RHS becomes $0$. For the first 2 derivatives, both LHS and RHS agree on 0.
    \end{proof}

\section{The Algorithm}\label{sec_the_algorithm}

Our algorithm is a sequence of estimations for each provided measurement. All estimations make use of all the copies of the given state $\rho$, and the individual procedures are non-adaptive, i.e., no classical information from the previous estimation steps is used in the subsequent steps.

\subsection{Main Idea} \label{intuition}

To better understand the algorithm, we try to provide some intuition. We consider the kind of noisy measurements used in \cite{aaronsonrothblum2019,badescu}, but instead of doing a threshold check, we try to measure out the noisy estimate directly. The usual way of thinking about this is to compute the sample mean coherently in an ancilla register, perform a weak measurement (corresponding to the exact noise one wants), and then uncompute the ancilla register. However, if the noise is additive, i.e., if the sample mean is distributed as a random variable $A$, we only want to sample from the distribution $A+X$ for some independent random variable $X$ corresponding to the noise, we can think of an alternate implementation. What follows is informal, and we do not present a formalization of this argument as we do not directly use it.

Instead of ancilla qubits, consider having access to a Quantum system akin to a particle in 1 dimension, i.e, the computational basis states of this system are $\ket x$ for $x\in \mathbb R$, and a quantum state in this register is of the form
\[\int_{-\infty}^{\infty} a(x)\ket x,\] where $a:\mathbb R\mapsto \mathbb C$ is a function with unit $L_2$ norm, i.e,

\[\int_{-\infty}^{\infty} a(x)\overline {a(x)}=1.\]
Now, corresponding to the noise $X$, initialize this register in the state
\begin{equation}
    \int_{-\infty}^{\infty} \sqrt{p_X(x)}\ket x\label{init_state_X},
\end{equation}

where $p_X$ is the p.d.f for $X$. Notice that if we directly measure this in the computational basis, the outcome is distributed as $X$. Next, let $M$ be the observable we want to measure with a noise. In the case of mean estimation, $M$ will be the observable corresponding to the sample mean $\frac 1 n \sum_{i=1}^n P^{(i)}$, where $P$ is the projector whose mean we wish to estimate. Next, let $T$ be the operator such that

\[\exp(\ii aT)\ket x=\ket{a+x}\]
for each $a,x\in \mathbb R$ (this is similar to the momentum operator in Quantum Mechanics). With this defined, we can sample from $A+X$ from the given copies of $\rho$ by first applying the following unitary operator
\[\exp(\ii M\otimes T),\]

(where the first tensor factor corresponds to the register containing the copies of $\rho$, and the second factor corresponds to the ancilla register), followed by measurement of the ancilla in the computational basis. This gives us a sample from $A+X$, because in the $\lambda$ eigenspace of $M$ (for some eigenvalue $\lambda$), the aforementioned unitary operator has action $\exp(\ii\lambda I\otimes T)=I\otimes \exp(\ii\lambda T)$, which simply translates the second register by $\lambda$.

For shadow tomography, we are interested in reusing the copies of $\rho$, so want to get an idea of what the resultant state looks like. Now, if the further procedure does not depend on the outcome observed in the ancilla register, for the analysis we can pretend that we did not see the outcome of the measurement, and that we have only traced out the register, which is further equivalent to measurement in any basis. In particular, we can choose the eigenbasis of the $T$ operator itself. What does this eigenbasis look like? It can be verified that it is indeed the Fourier basis, as the vectors

\[\int_{-\infty}^{\infty}\exp(\ii bx)\ket x\]

are easily seen to be $\exp(\ii ab)$ eigenvectors of $\exp(\ii aT)$, and thus are eigenvectors of $T$ with eigenvalue $b$. Now, since the applied unitary operator commutes with measurement in this basis, we can interchange the order of measurement and the unitary operation. So, we first measure the ancilla in the Fourier basis, and if the outcome obtained is $\lambda$, then the action of the unitary operator is the same as $\exp(\ii M\otimes \lambda I)=\exp(\ii \lambda M)\otimes I$, which can be thought of as a generalized form of phase kickback. Onto the first register, if $\lambda$ is small, the action is approximated by $I+\ii a\lambda M$. Also, the distribution of $\lambda$ depends only on the input state of the ancilla register, i.e, it only depends on the distribution $X$. So, if $\lambda$ is distributed symmetrically, then this action on average too is close to $I$. 

All of these arguments were only for a single estimation, and for multiple estimations, these phase kickbacks compose. The cumulative action of these phase kickbacks can be written, approximately, as $I+\ii \sum \lambda_i M_i$. Now, since all $M_i$'s are different, it is highly likely that with the observed $\lambda_{1}\dots \lambda_i$, the above action causes a large deviation in some direction, but if we look at the perturbation in any shadow of $\rho$, say $\mathbb E_\rho[M]$, then the perturbation approximately has the form $\sum \lambda_i c_i$ for some constants $c_i$ (each depending on $\rho,M$ and $M_i$). If each $\lambda_i$ has mean 0, we expect this quantity to be concentrated near $0$, and thus the shadows of $\rho$ will not be perturbed much.

In this setting, a good candidate for $X$ would be such that both the vector in Equation \ref{init_state_X} and its Fourier transform are close to 0 (i.e, the amplitudes are negligible for $\ket{x}$ with $|x|>>0$): the former because we want the estimation to be accurate, and the latter so that the phase kickbacks generated are not too large.  We also know that Heisenberg's Uncertainty principle puts a limit to this, as the product of the variances obtained would be $\Omega(1)$. Finally, we also know that the states 
\[\int_{-\infty}^\infty \exp\left(-\frac{x^2}{4\sigma^2}\right)\ket x\]

(after normalization), do indeed have the property that the product of the corresponding variances (in computational and fourier basis) is a fixed constant. This is because the Fourier transform of the standard Gaussian distribution is itself.

Since this ancilla register is infinite dimensional, we clearly cannot use it in an actual algorithm. We instead work with a collection of qubits, and the $\exp(\ii M\otimes T)$ operator here is replaced with $\exp(\ii M\otimes X)$  operator, where $X$ is the Pauli $X$ operator, which is then repeated for each extra qubit. Instead of translation, this applies a rotation onto the qubit, with angle proportional to the corresponding eigenvalue. In Section \ref{sec_low_mem} in the appendix, we discuss another algorithm that follows the idea discussed in this section much more closely.

\subsection{Estimation for $m=1$} \label{sec_single}

We first describe the procedure for a single $M$.

Given $n$ copies of state $\rho$, we use additional $k$ ancilla qubits ($k$ to be determined later). 

Each of the $k$ qubits is initialized in the state $\ket{\psi_0}= R_X\left(\frac 1 3 \pi\right)\ket 0=\frac {\sqrt{3}\ket 0 - \ii\ket 1} { 2} $.

Then, we apply the following unitary operator $n\cdot k$ times: once per pair of copy of state $\rho$ and ancilla qubit
\begin{equation}
    \exp\left(-\ii\frac \pi {6n}M\otimes  X\right)\label{M_rot}.
\end{equation}
Here, $X$ refers to the Pauli $X$ gate.
One can think of these as ``$M$-controlled'' rotations: if $M$ is diagonalized as $\sum_i \alpha_i\ketbra{\psi_i}{\psi_i}$, then the above acts as
\[\sum \ketbra{\psi_i}{\psi_i}\otimes R_X\left(\frac{\alpha_i \pi}{3n}\right).\]

The combined action of all the $n\cdot k$ unitary operators together is:

\begin{equation}
     \exp\left(-\ii\frac {\pi}{6n}\left(\sum_{i=1}^n M_{(i)}\right) \otimes\left(\sum_{i=1}^k X_{(i)}\right)\right)\label{u_all},
\end{equation}

where $M_{(i)}$ refers to the POVM element that is a tensor-product of $(n-1)$ $I_d$'s  and a single $M$, where the $M$ is the $i\ith$  tensor factor. Similarly $X_{(i)}$ is refers to a tensor-product of $(k-1)$ $I_2$'s and a single $X$, with $X$ being the $i\ith$ multiplicand.

After this, we measure out all the $k$ qubits in the computational basis. Suppose the output has $\mu$ fraction of 1's. Then, the estimate we output for $\mathbb E_{\rho}[M]$ is $\frac 6 \pi \cdot \arcsin (\sqrt \mu )-1$. 

Before moving on to the technical lemma, we can think of these $k$ qubits storing a noisy encoding of the estimated mean, and the final measurement is effectively a noisy measurement of the sample mean.

\begin{lemma}\label{lem_single_estimation}
    The above algorithm estimates $\mathbb E_\rho[M]$ up to $\epsilon\leq \epsilon_0$ additive error (where $\epsilon_0$ is some small enough absolute constant), with probability at most $\delta$ for $n\geq 10k$ and $k= \Theta\left(\frac {\log \frac 1 {\delta}}{\epsilon^2}\right)$ (where the hidden constant factor is large enough).    
\end{lemma}
\begin{proof}
    The correctness only depends on the statistics of the measurement of the $k$ ancilla qubits, so one can trace out the $n$ copies of state $\rho$ for analysis. We can presume that the $n$ copies were measured in the eigenbasis of $M$. Since the unitary operator \ref{M_rot} commutes with this measurement in the first register, we can WLOG assume that that measurement happened before the application of the $n\cdot k$ unitary operators.
    
    Suppose the $i\ith$ copy of state $\rho$ is observed to be in a state with eigenvalue $\lambda_i$ of $M$. Then, the effect of the $n\cdot k$ unitary operators is to apply the unitary operator $R_X(\frac \pi 3\cdot \frac 1 n\sum \lambda_i)$ on each of the $k$ ancilla qubits, which makes the state in each of them $R_X(\frac \pi 3 +\frac \pi 3\cdot \frac 1 n\sum \lambda_i)\ket 0$. Call the argument of $R_X$ as $\theta$.

    Now, by Hoeffding's bound, $|\frac 1 n\sum \lambda_i-\mathbb E[M]|<\epsilon$ with probability at least $1-\delta_1$. Note that $\theta\in \left(\frac \pi 3,\frac {2\pi} 3\right)$.

    On the other hand, we know that for each of the $k$ qubits, upon measurement, the probability of observing $\ket 1$ is $\sin^2{\frac \theta 2}$. So the fraction of 1s obtained $\mu$, with probability at least $1-\delta$, is at most $\epsilon$ far from $\sin ^2{\frac \theta 2}$. So by Lemma $\ref{lipschitz}$, $\arcsin(\sqrt \mu)$ is at most $\frac 2 {\sqrt 3} \epsilon $ far from $\frac \theta 2$, and so the estimate $\frac 6 \pi \cdot \arcsin (\sqrt \mu)-1 $ is at most $\frac {12}{\pi \sqrt 3} \epsilon <3\epsilon $ away from $\frac 1 n \sum \lambda_i$. (Note that here we need that $\epsilon$ is small enough so that Lemma $\ref{lipschitz}$ is applicable).

    Combining the above two results, we get that the estimate is almost $4\epsilon$ away from the true expectation, with probability at least $1-2\delta$. Scaling $\epsilon$ and $\delta$  in the proof appropriately gives us the required result.    
\end{proof}
\subsection{Generalised algorithm for $m$ POVM's}
    The general algorithm is to sequentially use the aforementioned algorithm, reusing all the copies of state $\rho$, but with a fresh of $k$ ancilla qubits. Now, each step does alter the state, but we show that at any step, the sample mean w.r.t any POVM element $M$ is not altered much.

\begin{algorithm}
\label{the_algorithm}
\caption{Algorithm for shadow tomography}
\begin{algorithmic}[1]
    \STATE Register $A\gets \rho^{\otimes n}$
    \STATE $k \gets c_0\frac 1 {\epsilon^2}\log\frac 1 \delta$
    \FOR{$j = 1$ to $m$}
        \STATE Initialize $k$ qubits $B=b_1b_2\dots b_k$ in state $\ket {\psi_0}^{\otimes k}$
        \STATE Apply unitary  $\exp\left(-\ii\frac {\pi}{6n}\left(\sum_{i=1}^n M_{j,(i)}\right) \otimes\left(\sum_{i=1}^k X_{(i)}\right)\right)$ onto registers $AB$.
        \STATE Measure $B$ in computational basis.
        \STATE $\mu\gets$ fraction of 1's in the output
        \STATE Output  $\frac 6 \pi \cdot \arcsin (\sqrt \mu )-1$
    \ENDFOR
\end{algorithmic}
\end{algorithm}
    Before we jump into the analysis, we describe the general idea:
    \begin{enumerate}
        \item We will focus on showing the correctness for one POVM at a time. The way we do this is as follows: for a fixed POVM $M_i$, we see how the estimation for the $j\ith$ POVM element $M_j$ affects the estimate for $M_i$ if we were to directly do it after $M_j$.
        \item In fact, we do something stronger: we keep track of exactly what the resultant state will be after the $j\ith$ estimate, if the $j\ith$ step was done slightly differently, which wouldn't affect the correctness of the remaining estimates.
        \item In this modified setting, the resultant state will turn out to also be $n$ copies of some state $\rho'$ close to the original state $\rho$.
        \item We use this to bound the perturbation in the expected value. 
    \end{enumerate}
      \subsection{The perturbation of $\rho$}\label{sec_perturbation}

      Suppose at the beginning of the $j\ith$ estimation procedure, we have $n$ copies of some state $\rho_j$, i.e., we have $\rho_j^{\otimes n}$, and we run the estimation procedure mentioned earlier mentioned for $M_j$ with $k$ qubits.  

      The unitary operation applied, as mentioned before, is

      \[U_j=\exp\left(-\frac {\ii\pi}{6n} \left(\sum_{i=1}^n M_{j,(i)}\right)\otimes \left(\sum_{i=1}^k X_{(i)}\right)\right).\]

        Here, $M_{j,(i)}$ is the tensor product of matrices which are identity at all but the $i\ith$ place, where it is $M_j$.
      To obtain the $j\ith$ estimate, we had measured the $k$ qubits in the computational basis. Now, if we are only concerned with the correctness of estimates after the $j\ith$ one, we only need to track the mixed state obtained after the computational basis measurement, i.e., we can presume we did not look at the $j\ith$ measurement outcome. However, for this, it is equivalent to measuring the $k$ qubits in any basis. In particular, what we can do is instead measure the $k$ qubits in the basis in which $\left(\sum_{i=1}^k X_{(i)}\right)$ is diagonalized, which we know is the $\ket \pm $ basis. Equivalently, we can think of measuring $\left(\sum_{i=1}^k X_{(i)}\right)$ as an observable, which will give us one of the eigenvalues.  Note that we are doing this only for analysis and that this measurement doesn't actually occur in the algorithm. 

      Note now, that this measurement onto the $k$ qubits commutes with $U_j$, i.e., we could have instead performed this measurement before applying $U_j$, and the resultant state will not change. So, suppose, before applying $U_j$ , we get $\lambda_j$ as the measurement outcome when we measure $\left(\sum_{i=1}^k X_{(i)}\right)$ as an observable on the initial state $\ket{\psi_0}^{\otimes k}$ of the $k$ qubits. (We will discuss the distribution of $\lambda_j$ later.)

      Knowing $\lambda_j$, we know that the action of $U_j$ is equivalent to the action of 
      \[U_{j,\lambda_j}=\exp\left(-\frac {\ii\pi}{6n} \left(\sum_{i=1}^n M_{j,(i)}\right)\otimes (\lambda_j I)\right).\]
      as we have already projected onto the subspace where $\left(\sum_{i=1}^k X_{(i)}\right)$ has eigenvalue $\lambda_j$. Note that this new action does not interact with the $k$ qubits anymore, so we can simply look at the action onto the $n$ copies of $\rho_j$, which is
      \[\exp\left(-\frac {\ii\pi\lambda_j}{6n} \left(\sum_{i=1}^n M_{j,(i)}\right)\right)=\exp\left(-\frac {\ii\pi\lambda_j}{6n}  M_{j}\right)^{\otimes n}.\]
      Upon action of this unitary operator, we have
      \[\rho_j^{\otimes n}\mapsto \left(e^{-\frac {\ii\pi\lambda_j}{6n}M_j}\rho_j e^{\frac {\ii\pi\lambda_j}{6n}M_j}\right)^{\otimes n}.\]
      
      So, given $\lambda_j$, the resultant state still has a product structure, and so the algorithm can be used once more on the same copies to estimate the expectation of another POVM in a similar fashion, but the estimates we will get are for this new state $\rho_{j+1}=e^{-\frac {\ii\pi\lambda_j}{6n}M_j}\rho_j e^{\frac {\ii\pi\lambda_j}{6n}M_j}$. So, to show that the algorithm will still work after this, we need to ensure
      \begin{enumerate}
          \item Given $\lambda_j$, the expectation for the new state is not far off from the original
          \item The aforementioned disturbance needs to then take into account that $\lambda_j$ is random, so we instead need concentration of the disturbances.
      \end{enumerate}

      \subsection{Distribution of $\lambda_j$}\label{sec_eigen_dist}

      As we said earlier, $\left(\sum_{i=1}^k X_{(i)}\right)$ is diagonalised in the $\ket \pm $ basis. So, to get $\lambda_j$, we can simply measure each $X_{(i)}$ and sum up the outcomes. Now, on the initial state $\ket{\psi_0}$ , measuring $X_{(i)}$ leads to outcomes $+1$ and $-1$ with equal probability. So, $\lambda_j$ is distributed as the sum of $k$ Bernoulli $\pm 1$ variables. 
     \subsection{Bounding deviation in expectation}\label{sec_bnd_dv}

      In the $j\ith$ round, modified as described above, the initial state was $\rho_j^{\otimes n}$, and the final state is ${\rho_{j+1}}^{\otimes n}=\left(e^{-\frac {\ii\pi\lambda_j}{6n}M_j}\rho_j e^{\frac {\ii\pi\lambda_j}{6n}M_j}\right)^{\otimes n}$. So, the new expectation of $M_i$ is
      \begin{align*}
          \Tr[M_ie^{-\frac {\ii\pi\lambda_j}{6n}M_j}\rho_{j} e^{\frac {\ii\pi\lambda_j}{6n}M_j}]&=  \Tr[e^{\frac {\ii\pi\lambda_j}{6n}M_j}M_ie^{-\frac {\ii\pi\lambda_j}{6n}M_j}\rho_j ].
      \end{align*}

      Using Lemma \ref{lem_conjugate_bound}, and that $\norm{\frac {\pi \lambda_j}{6n}M_j}\leq 1$ (as $|\lambda_j|\leq k \leq n/10$), we get that 

      \begin{align*}
           &\quad \left|\Tr[e^{\frac {\ii\pi\lambda_j}{6n}M_j}M_ie^{-\frac {\ii\pi\lambda_j}{6n}M_j}\rho_j ]-\Tr[M_i\rho_j]-\frac{\pi\lambda_j}{6n}\Tr[\ii[M_j,M_i]\rho_j] \right|\\
           \leq &\quad \frac{\pi^2\lambda_j^2}{36n^2}\norm{[M_j[M_j,M_i]]}\cdot \frac{e^2-3}{4}\\
           \leq & \quad\frac{\pi^2\lambda_j^2}{36n^2}\norm{[M_j,M_i]}\cdot \frac{e^2-3}{2}\\
           \leq &\quad c_2\frac{\lambda_j^2}{n^2}
      \end{align*}
for some absolute constant $c_2$. Using this, we can define the third  term inside the $|\circ|$ as $S_{1,j}\eqdef\frac{\pi\lambda_j}{6n}\Tr[\ii[M_j,M_i]\rho_j]=c_1\frac {\lambda_j}{n}\Tr[\ii[M_j,M_i]\rho_j] $ for absolute constant $c_1=\frac{\pi\lambda_j}{6n}$. We can then write 
\begin{equation}
    \Tr[M_i\rho_{j+1}]-\Tr[M_i\rho_j]=S_{1,j}+S_{2,j},\label{eqn_deviation}
\end{equation}
with $S_{2,j}\leq  c_2\frac{\lambda_j^2}{n^2}$.

     \subsection{Combining everything together}\label{sec_cmb_all_tog}

\begin{theorem}\label{thm_sample_cmpl}
    The aforementioned estimation procedure, when used sequentially for each $M_j$, outputs estimates for each $\mathbb E_\rho[M_j]$, which, for any chosen index $i$, has an additive error at most $\epsilon$, with probability at least $1-\delta$, when the following constraints are satisfied.
    \begin{itemize}
        \item $k=c_0\frac 1 {\epsilon ^2}\log\frac 1{\delta}$.
        \item $n\geq 10k$.
        \item $n^2\geq 8c_1^2mk\frac 1 {\epsilon^2}\log \frac 1 \delta  $.
        \item $c_2C^2\frac{mk}{n^2}\ln 2+c_2C^2\frac k {n^2}\ln \frac 1 \delta\leq \epsilon $.
    \end{itemize}
    Here, $c_0,c_1,c_2$, and $C$ are all absolute constants.
\end{theorem}
This suffices to show the required sample complexity by the following simple union-bound argument:
\begin{corollary}\label{cor_main}
    The estimation procedure, with $n=O\left(\frac {\sqrt m}{\epsilon^2}\log \frac m \delta \right)$ samples of $\rho$ simultaneously estimates all $\mathbb E_\rho[M_j]$ with additive error at most $\epsilon$, with probability at least $1-\delta$.
\end{corollary}
    \begin{proof}
        It is easy to see that the constraints in Theorem $\ref{thm_sample_cmpl}$ can be satisfied with $n=O\left(\frac {\sqrt m}{\epsilon^2}\log \frac 1 \delta \right)$. Now, invoke the theorem with $\delta$ replaced with $\frac {\delta}{m}$. Then, the sample complexity becomes $n=O\left(\frac {\sqrt m}{\epsilon^2}\log \frac m \delta \right)$, and by the union bound, the overall probability of failure is at most $\delta$.
    \end{proof}

We move on to the proof of Theorem \ref{thm_sample_cmpl}.

\begin{proofof}{Theorem \ref{thm_sample_cmpl}}
Fix an index $i\in [m]$. Our goal is to show that the $i\ith$ estimate is within $\epsilon$ distance of $\mathbb E_\rho[M_i]$. We will actually prove this up to constant factors in $\epsilon$ and $\delta$, which suffices, as we can correspondingly divide $\epsilon$ and $\delta$ by appropriate factors, only affecting the values of the absolute constants.
    % $n\geq \sqrt{c_0  m}\ln \frac 1{\delta} \frac 1 {\epsilon^2}$ (to be calculated), and $k=c_0\frac 1 {\epsilon ^2}\log\frac 1{\delta}$, the $i\ith$  estimate is within $\epsilon$ of the true expectation with probability atleast $1-O(\delta)$.

    Fix $k=c_0\frac 1 {\epsilon ^2}\log\frac 1{\delta}$ such that any single estimation goes through with probability $1-\delta$ in case of a single estimation (assuming $n\geq 10k$), using Lemma \ref{lem_single_estimation}. So, we just need to bound the deviation in mean after $i-1$ to within $O(\epsilon)$ with probability $1-O(\delta)$.

    The $i\ith$ measurement outcome can only be affected by the previous $i-1$ measurements. As before, we will, for the sake of analysis, assume that the measurements were done in the $\ket \pm$ basis, effectively measuring the observable $\sum_{l=1}^{k-1} X_{(l)}$ at each step, leading to the observed value $\lambda_j$ at the $j\ith$ step, and leading to a new product state as described earlier. We keep track of this change in the following way:

     Define the following intermediate quantum states: $\rho_1=\rho$, and 
     \[\rho_{j+1}=e^{-\frac {\ii\pi \lambda_j}{6n}M_j}\rho_{j}e^{\frac {\ii\pi \lambda_j}{6n}M_j}.\]
     i.e. , $\rho_{j+1}$ is the quantum state for which we have $n$ copies just before the $j\ith$ estimation step.

     With this, we can now write 
     \begin{align}
         |\Tr[M_i\rho_i]-\Tr[M_i\rho]|&=\left|\sum_{j=1}^{i-1}\left(\Tr[M_{i}\rho_{j+1}]-\Tr[M_i\rho_j]\right)\right|\label{eqn_telescoped_bgn}\\
         &=\left|\sum_{j=1}^{i-1}\left(S_{1,j}+S_{2,j}\right)\right|&\text{[as in \ref{eqn_deviation}]}\\
         &\leq \left|\sum_{j=1}^{i-1}S_{1,j}\right| + \left|\sum_{j=1}^{i-1}S_{2,j}\right|\\
         &\leq \frac {\pi}{6n}\left|\sum_{j=1}^{i-1}\lambda_j \Tr[i[M_i,M_j]\rho_j]\right| + \frac {\pi^2}{9n^2}\sum_{j=1}^{i-1}\lambda_j^2 \\
         &= \frac {c_1}{n}\left|\sum_{j=1}^{i-1}\lambda_j \Tr[i[M_i,M_j]\rho_j]\right| + \frac {c_2}{n^2}\sum_{j=1}^{i-1}\lambda_j^2 . \label{eqn_telescoped_difference}
     \end{align}
     In the first term, the quantity inside $|\circ |$ is a martingale difference sequence. It is easy to see that the expectation of quantity is 0, as each $\lambda_j$ is symmetric, and each $\rho_j$ only depends on the $\lambda_j$'s before them. We now prove that the quantity is indeed concentrated near 0 by using Azuma's inequality. For simplicity, define:
     \[q_j= \Tr[i[M_i,M_j]\rho_j]\]
     Note that $|q_j|\leq 2$, as $\norm{i[M_i,M_j]}\leq 2\norm{M_iM_j}\leq 2$, and as with $\rho_j$, $q_j$ is only dependent on $\lambda_l$s' for which $1\leq l<j-1$, which we denote by the vector $\lambda_{1:j-1}$

Note that the expression $\sum_{j=1}^{i-1} \lambda_jq_j$ can be rewritten as $\sum_{j=1}^{i-1} \sum_{l=1}^k\lambda_{j,l}q_j$, where each $\lambda_{j,l}$ is a uniform Bernoulli $\pm 1$ random variable. This is still a martingale difference sequence, but now, since each $\lambda_{j,l}q_j$ is bounded and has 0 mean (even conditioned on previous ones), we can apply Azuma's inequality (Lemma \ref{lem_azuma}) to get: 

\begin{align}
\Pr\left[\frac {c_1}{n}\left|\sum_{j=1}^{i-1}\lambda_j \Tr[i[M_i,M_j]\rho_j]\right|\geq  \epsilon \right]&=
  \Pr\left[\left|\sum_{j=1}^{i-1}\lambda_j \Tr[i[M_i,M_j]\rho_j]\right|\geq \frac {n}{c_1}\cdot \epsilon \right]\\
  &\leq 2\exp\left(-\frac{\epsilon^2n^2}{2mkc_1^2\norm{[M_i,M_j]}^2}\right)\label{eqn_first_order_prob}\\
  &\leq 2\exp\left(-\frac{\epsilon^2n^2}{8mkc_1^2}\right)  
\end{align}

        By the given constraints, this is at most $2\delta$.

        Now, for the other term, we need to just find concentration bounds for
        \[\frac {c_2}{n^2}\sum_{j=1}^{i-1} \lambda_j^2.\]
         For this, notice that each $\lambda_j$ is a sum of $k$ independent Bernoulli random variables. So, as a random variable, their Sub-Gaussian norms are at most $C\sqrt k$ for some absolute constant $C$, by Lemma \ref{subg_sum}.
        
        So, we have that 
        \[\mathbb E\left[\exp\left(\frac{\lambda_j^2}{C^2k}\right)\right]\leq 2\]
        \[\implies\mathbb E\left[\exp\left(\frac{\sum_{j=1}^{i-1} \lambda_j^2}{C^2k}\right)\right]\leq 2^{(i-1)}. \]
        From this, we get the concentration bound
        \[\Pr\left[\sum_j^{i-1}\lambda_j^2\geq \alpha\right]=\Pr\left[\exp((C^2k)^{-1}\sum_j^{i-1}\lambda_j^2)\leq e^{\frac{\alpha}{C^2k}}\right]. \]
        \[\leq \frac {2^i}{e^{\frac{\alpha}{C^2k}}}\leq \frac {2^m}{e^{\frac{\alpha}{C^2k}}}. \]
        For $\alpha=mC^2k\ln 2+C^2k\ln\frac 1 \delta$, we get
        \[\Pr\left[\sum_j^{i-1}\lambda_j^2\geq mC^2k\ln 2+C^2k  \ln \frac 1 \delta\right]\leq \delta.\]
        From this, we get,
        \[\Pr\left[\frac {c_2}{n^2} \sum_{j=1}^{i-1}\lambda_j^2\geq c_2C^2\frac{mk}{n^2}\ln 2+c_2C^2\frac k {n^2}\ln \frac 1 \delta\right]\leq \delta.\]
        By the given constraints, we know that the threshold is at most $\epsilon$. So, combining all these bounds, we get that the $i\ith$ estimate is within $3\epsilon$ additive error of $\mathbb E_\rho[M_i]$ with probability at least $1-4\delta$. 
\end{proofof}
\subsection{Other properties}\label{sec_other_props}

Having established the sample complexity in Theorem \ref{main_thm}, we now move on to the other claimed properties of the algorithm. 

\paragraph{Trivial Classical post-processing :} After the application of the quantum circuit and the subsequent measurements, we just need to compute the estimate $\frac 6  \pi \cdot \arcsin(\sqrt \mu) -1$ for each stage, so post the application of the quantum circuit no sophisticated classical processing is needed. 

\paragraph{Robust against qubit measurement noise:} We can assume that the measurements are performed after all gates are applied. So, measurement noise can only disturb the fraction of 0's and 1's we obtain for each estimation, and so it only affects the accuracy of the estimate. The perturbation in the estimate is proportional to the number of bit flips the noise causes. Furthermore, if this noise is well characterized, then the correct estimate may as well be extracted from the noisy estimate.

\paragraph{Gate Complexity:} We talk about the gate complexity in the case where the POVMs are efficiently implementable. Any implementation of a POVM as a quantum circuit corresponds to a projective measurement defined on the register plus some ancilla qubits, originally all in the state $\ket 0$. Since none of our arguments have any dimension dependence, we can, without loss of generality, assume that we are only dealing with projectors.

 We only need to implement the unitary operator in Equation $\ref{u_all}$, which is \begin{equation}
     \exp\left(-\ii\frac {\pi}{6n}\left(\sum_{i=1}^n M_{(i)}\right) \otimes\left(\sum_{i=1}^k X_{(i)}\right)\right),
\end{equation} $m$ times, where $M$ will vary over all the $M_i$'s. In each stage, we can measure each copy of $\rho$ with the projector $M$, and store the observations in $n$ corresponding ancilla qubits. We denote these $n$ qubits together as $P_i$, and the $j\ith$ one as $P_{i,j}$. This leads to gate complexity $n\cdot S_i$ at stage $i$. Next, we apply controlled rotations $CR_X(\frac {\pi}{3n})$, controlled from each of the qubits in $P_i$, to each of the $k$ ancilla qubits prepared in the state $\ket{\psi_0} $. The total gate complexity for this step is $n\cdot k$. Finally, we uncompute $M_i$ from all qubits in $P_i$.

Repeating this for all $m$ Projectors, we get the desired gate complexity of $\Theta(m\cdot k\cdot n+n\cdot \sum S_i)$. 

\paragraph{Implementation as a read-once quantum circuit with low memory :} Notice that the $P_{i,j}$s can be uncomputed directly after the application of all the corresponding controlled rotations from it to the $k$ ancilla qubits. So, we can modify the algorithm such that immediately after the computation of $P_{i,j}$, we apply all the corresponding controlled rotations ($k$ of those), and uncompute $P_{i,j}$. Notice that 
\begin{enumerate}
    \item $P_{i,j}$ can be discarded after this, and no two $P_{i,j}$'s are computed simultaenously, i.e, $P_{i,j}$ is uncomputed before any other $P_{i',j'}$ is computed. Therefore, we can keep reusing the same qubit for all $P_{i,j}$'s.
    \item The next operation involving the $j\ith$ copy of $\rho$ is computation of $P_{i+1,j}$. So, we can reorder the circuit so that all operations corresponding to the $j\ith$ copy of $\rho$ (for all $M_i$'s together) happen in one go.
\end{enumerate}

With this modification, we end up using only $\Theta(mk+n)$ qubits, and the circuit only accesses one copy of $\rho$ at a time.

\paragraph{Implementation with $O(1)$ additional memory:} Instead of applying the controlled rotation to the $k$ qubits together, we can perform it for a single qubit, out of $k$, at a time, measure it, and then repeat for each of the remaining $k-1$ qubits. Because of this, now a single qubit can be reused in all the $k$ rounds (assuming access to reset gates), and as before, we are using a single qubit for all the $P_{i,j}$'s as well. However, notice that because of this, each $P_{i,j}$ now needs to be computed (and uncomputed) $k$ times. Combining the effect over all the $m$ POVMs, the gate complexity has an additional overhead of $\Theta(nk\sum S_i)$, but now the algorithm only requires $2$ qubits. In general, with this approach, with $s+1$ qubits, the gate complexity overhead is $\Theta(\frac{nk}{s}\sum S_i)$.

\paragraph{Implementation with $O(\log n)$ extra memory:} With extra $O(\log n)$ bits, in each round (corresponding to the estimation of a particular $\mathbb E_{\rho}[M_i]$), we can sum up the answers of all $P_{i,j}$ in a size $O(\log n)$ register using the QFT based adder \cite{draper2000}. Notice that we can do this one at a time: compute $P_{i,j}$, add to the register, and uncompute $P_{i,j}$. Upon computation of the sum, we can apply multiplicity controlled rotations from this register to the $k$ ancilla qubits, one at a time, so we can reuse them. That is, we prepare one qubit, apply multiplicity-controlled rotations, measure it, reset it, and then repeat the process $k$ times on the same qubit.

Repeating this for each of the $m$ rounds (reusing the addition register for all rounds), we use only $O(\log n)$ additional memory and incur a smaller gate complexity overhead of $\Theta(mn\log n +mk\log n)$. Here, the $\log n$ factor appears due to the complexity of addition and multiplicity-controlled rotations.

These are just a few ways one can modify the algorithm, as it can be restructured in many other ways by reordering the computations and uncomputations, and shifting around commuting operations.
\paragraph{Other properties:} Note that our analysis considers the worst case. So, for many choices of the POVM's $M_1\dots M_m$, the inequalities will not be tight, which means that in certain cases, our techniques might lead to a smaller sample complexity for our algorithm. Consider the case when all the commutators have their norm bounded by some $C_{max}$. Then, if take this into account while using \ref{lem_conjugate_bound} to bound the deviations in Section \ref{sec_bnd_dv}, we get that $S_{2,j}$ in Equation \ref{eqn_deviation} is instead upper bounded by $c_2C_max\frac{\lambda_j^2}{n^2}$ (we allow changing the absolute constant $c_2$). We can also directly we can upper bound the expression in Equation \ref{eqn_first_order_prob} by 
\[2\exp\left(-\frac{\epsilon^2n^2}{2mkc_1^2C_{max}^2}\right).\]  With these changes, we can replace the constraints in Theorem \ref{thm_sample_cmpl} by:
\begin{itemize}
        \item $k=c_0\frac 1 {\epsilon ^2}\log\frac 1{\delta}$.
        \item $n\geq 10k$.
        \item $n^2\geq 2c_1^2C_{max}^2mk\frac 1 {\epsilon^2}\log \frac 1 \delta  $.
        \item $c_2C^2\frac{mk}{n^2}C_{max}\ln 2+c_2C^2C_{max}\frac k {n^2}\ln \frac 1 \delta\leq \epsilon $,
\end{itemize}

 which implies that the sample complexity for this special case (for the requirement in Theorem \ref{thm_sample_cmpl}) is $n=O\left( (C_{max}\sqrt m+1)\frac{1}{\epsilon^2} \log \frac {1}\delta+ \sqrt{C_{max}}\frac{\sqrt{m\log (1/\delta)}+\log(1/\delta)}{\epsilon^{1.5}} \right)$, and the true sample complexity, as described in Corollary \ref{cor_main}, is obtained by replacing $\delta$ with $\frac {\delta}{m}$.

\bibliographystyle{plain}
\bibliography{bibl}
\appendix
\section{Low memory read-once Quantum Circuit}\label{sec_low_mem}

In this section, we discuss how to bring down the number of ancilla qubits used in each estimation down to $O(\log n)$, which leads to a read-once Quantum Circuit with $O(m\log n)$ ancilla qubits.

WLOG let the number of samples $n$ be a power of $2$, for simplicity of the argument. Also, for this section, we restrict to the case of Projectors instead of arbitrary POVM elements, which by Naimark's dilation theorem, is sufficient with respect to sample complexity. For computational efficiency, we note that implementations of POVMs can again be modified to give a Projective measurement with the same gate complexity.

We want to maintain the guarantees of \ref{thm_sample_cmpl} (up to changes in the absolute constants). 

Consider an ancilla register $B$ with $2\log n +2$ qubits, in which we can store integers from $-2n^2$ to $2n^2-1$, i.e., the computational basis states can be labeled as $\ket x$ for $x$ being an integer in that range. Call this range as $R$. Define $N\eqdef 2n^2$ so that $R=[-N,N-1]\cap \mathbb Z$.

In this register, consider the initial state

\[\ket{\psi_p}=\frac {1}{\sqrt{\binom {2p}{4p}}} \sum_{x\in R}  \binom {2p}{p+x} \ket x, \]

for some $p$, to be suitably chosen later. We postpone showing that this is normalized.  Using the idea described in Section \ref{intuition}, we can use the state to sample the sum with the corresponding noise, provided that we can define a corresponding addition operator. We will discuss this noise later. For now, call the distribution $S$.

Now, consider the operator $Q$ which has the action

\[Q\ket x= Q\ket{x+1\pmod{ 2N}}.\]

This can easily be implemented using $CNOT, X$, and Toffoli gates. Define $T$ such that it satisfies $\exp(-\ii T)=Q$ and has the smallest possible eigenvalues (in magnitude). Since $Q^{4N}=I$, we have that all eigenvalues of $T$ are of the form $\frac {a{\pi}}{2N}$, where $a\in R$. Alternatively, we can also see this by seeing that, similar to Section \ref{intuition}, the Fourier vectors form the eigenbasis of $T$.

Then, if we want to estimate the mean of projector $M$ using $n$ copies of $\rho$ in register $A$,  we initialize the register $B$ in the given state, and apply the following unitary operator on $AB$:

\begin{equation}
    \exp\left(-\ii \left(\sum_{i=1}^n M^{(i)}\right)\otimes T  \right).\label{unitary_disc}
\end{equation} 
As in Section \ref{sec_single}, these are just $M$ controlled translations from each of the copies of $\rho$ to the ancilla register $B$. Note that since $M$ is a projector, we do not get ''halfway rotations", i.e., if the second register is in a computational basis state, it remains in one after this operation (assuming $A$ has already been measured in the eigenbasis of $M$).

After this, as before, we measure the register $B$ in the computational basis. Suppose the observation is $\mu$. The estimate we output is $\frac \mu n$. Using this, we can obtain the following lemma, analogous to Lemma \ref{lem_single_estimation}:

\begin{lemma}\label{lem_low_mem_single} The aforementioned estimation procedure outputs an estimate of the sample mean $\mathbb E_\rho[M]$ with additive error at most $\epsilon$, with probability at least $1-\delta$, when $p=O\left(\frac {n^2\epsilon^2} {\log(1/\delta)}\right)$, and $n=\Omega\left(\frac {\log (1/\delta)}{\epsilon^2}\right)$
\end{lemma}
\begin{proof}
    As described in Section \ref{intuition}, we are effectively sampling from the random variable $A+S$, where $A$ is distributed as the empirical mean of $M$ on $\rho^{\otimes n}$. As before, by Hoeffding's inequality, we have that $A$ is within $n\epsilon$ of its mean with probability at least $1-\delta$. So, it remains to show that $S$ is also $n\epsilon$ close to $0$ with probability at least $1-\delta$.

    We first show that $\ket{\psi_p}$ is normalized. The square of the amplitudes are $\binom{p+x}{2p}^2 \cdot \binom{2p}{4p}^{-1}$. Now, we know that 

    \[\sum_{x\in R}\binom{2p}{p+x}^2=\sum_{i=0}^{2p} \binom{2p} {i}^2=\binom {4p}{2p},\]

    which shows the normalisation of $\ket{\psi_p}$. This actually shows something stronger: Consider randomly sampling $4p$ i.i.d. uniform Bernoulli 0,1  random variables.  Divide it into 2 sections of $2p$ random variables. Now, condition on the total sum being $2p$. Then, the conditional probability of getting exactly $p+x$ 1s in the first half is $\binom{p+x}{2p}\cdot \binom{p-x}{2p}\cdot \binom{2p}{4p}^{-1}= \binom{p+x}{2p}^2 \cdot \binom{2p}{4p}^{-1}$. This matches up with the distribution of $S$. 

    Now, also observe that this distribution is identical to the one where we have a bit string with $2p$ 1s and $2p$ 0s, and we sum up the elements at $2p$ uniformly random distinct positions (without replacement). In \cite{Hoeffding1963}, Hoeffding showed that the bound of Lemma \ref{lem_hoeffding} also holds in this setting.  So, using $t=n\epsilon$ in Lemma \ref{lem_hoeffding}, we get that

    \[\Pr[|S|\geq n\epsilon]\leq 2 \exp\left(-\frac{2n^2\epsilon^2}{8p}\right)\leq \delta \]

    By sufficiently scaling $\epsilon$ and $\delta$ in this argument, we get the required result.
\end{proof}

We now want to show the analogous result in the remaining subsections of Section \ref{sec_the_algorithm}. For this, we rewrite the unitary in Equation \ref{unitary_disc} as 

\[\exp\left(-\ii \frac 1 n\left(\sum_{i=1}^n M^{(i)}\right)\otimes (n\cdot T)  \right),\]

to express it in a form similar to Equation \ref{u_all}.

As in Algorithm \ref{the_algorithm}, the new algorithm is to simply repeat the estimation for all measurements $M_i$.

\begin{algorithm}
\label{the_algorithm_low_mem}
\caption{Lower memory algorithm for shadow tomography}
\begin{algorithmic}[1]
    \STATE Register $A\gets \rho^{\otimes n}$
    \STATE $N\gets 2n^2$
    \STATE $p \gets c_0\frac 1 {\epsilon^2}\log\frac 1 \delta$
    \FOR{$j = 1$ to $m$}
        \STATE Initialize register $B$ (with $2\log n+2$ qubits) in state $\ket{\psi_p}$.
        \STATE Apply unitary  $\exp\left(-\ii\left(\sum_{i=1}^n M_{j,(i)}\right) \otimes T \right)$ onto registers $AB$.
        \STATE Measure $B$ in computational basis. Suppose observed $\mu$
        \STATE Output  $\frac {\mu}{n}$
    \ENDFOR
\end{algorithmic}
\end{algorithm}

Next, noting the rewritten form of Equation \ref{unitary_disc} above, we retrieve the form of the perturbation in Section \ref{sec_perturbation}, just this time the eigenvalues $\lambda_j$ are of the operator $n\cdot T$ instead, and the Fourier basis measurement has been performed on the new initial state $\ket{\psi_p}$.

Unlike Section \ref{sec_eigen_dist}, we need to put in a bit more work to figure out the distribution corresponding to the Fourier transform of $\ket{\psi_p}$. Before doing that, observe that the Fourier basis vector 

\[\sum_{l\in R} \exp\left(\ii  \frac {2\pi l}{2N}\right)\ket x\]

is a $\frac {2\pi l}{2N}$ eigenvector of $T$, and therefore is a $\frac{\pi l}{2n}$ eigenvector of $n\cdot T$. 

Now, the Fourier transform of $\ket {\psi_p}$ is

\begin{align}
    F\ket{\psi_p}&=\frac 1 {\sqrt{2\binom{2p}{4p}N}}\sum_{j\in R}\left(\sum_{l\in R}\exp\left(\ii \frac {2\pi lj}{2N}\right)\binom {p+l}{2p} \right)\ket j\\ 
    &=\frac 1 {\sqrt{2\binom{2p}{4p}N}}\sum_{j\in R}\left(\exp\left(\ii \frac{j\pi }{2N}\right)+\exp\left(-\ii \frac{j\pi }{2N}\right)\right)^{2p}\ket j\\ 
    &=\frac 1 {\sqrt{2\binom{2p}{4p}N}}\sum_{j\in R}\left(2\cos\left(\frac{j\pi }{2N}\right)\right)^{2p}\ket j
\end{align}

Hence, on Fourier basis measurement, we observe the $l\ith$ Fourier basis vector (for $l\in R$) with probability $f(l)=\frac {\left(2\cos\left(\frac{l\pi}{2N}\right)\right)^{4p}}{2\binom {2p}{4p}N }$. Let $L$ be a random variable with this distribution, then $\lambda_j\sim \frac {\pi L}{2n}$.
\subsection{Sub-Gaussian Norm}

We now upper bound the Sub-Gaussian norm of the described distribution $L$, which directly bounds the Sub-Gaussian norm of $\lambda_j$. We know that if 

\begin{equation}
  \mathbb E[\exp(L^2/K^2)]\leq 2,  \label{eqn_L_sbgnorm}
\end{equation}

then $\norm{L}_{\psi_2}\leq K$ and $\norm{\lambda_j}_{\psi_2}\leq K \cdot \frac {\pi}{2n}$. We show that this holds with some $K=O\left(\frac{N}{\sqrt p}\right)$. First, observe that $f(x)=f(0)\cdot \cos^{4p}\left(\frac{x\pi}{2N}\right)$. Next, take some $l_{max}\in R\cap \mathbb Z_+$. We will fix this value later on.  

We first bound $f(0)=\frac {2^{4p}}{2\binom {4p}{2p}N}$. By Sterling's approximation, we know that $f(0)= \Theta\left(\frac{\sqrt p}{N}\right)$. We use the more explicit bound of $f(0)\leq \frac{\sqrt{2p}}N$, which can be derived from the bound in \cite{stack_exchange_sterling}. With this, we have

\begin{align}
    \mathbb E[\exp(\lambda_j^2/K^2)]&\leq \exp\left(\frac {l_{max}^2}{K^2}\right)+ f(0)\sum_{l\in R, |l|\leq l_{max} } \cos^{4p}\left(\frac{l\pi}{2N}\right)\cdot \exp\left(\frac {l^2}{K^2}\right) &[\text{Splitting at $l_{max}$}]\\
    &\leq  \exp\left(\frac {l_{max}^2}{K^2}\right)+ f(0)\sum_{l\in R, |l|\leq l_{max} } \left(1-\frac 1 4\left(\frac{l\pi}{2N}\right)^2\right)^{4p}\cdot \exp\left(\frac {l^2}{K^2}\right)&[\text{by \ref{lem_cos_bound}}]\\
    &\leq  \exp\left(\frac {l_{max}^2}{K^2}\right)+ f(0)\sum_{l\in R, |l|\leq l_{max} } \exp\left(\frac {l^2}{K^2}-p\left(\frac{l\pi}{2N}\right)^2\right)\\
    &=\exp\left(\frac {l_{max}^2}{K^2}\right)+ f(0)\sum_{l\in R, |l|\leq l_{max} } \exp\left(l^2\left(\frac {1}{K^2}-p\left(\frac{\pi}{2N}\right)^2\right)\right)\\
\end{align}

Let $K=\frac {4N}{\pi \sqrt p}K_0$, for some $K_0\geq 1$ (to be fixed later), from which we get $\frac 1{K^2}\leq \frac{p}2\left(\frac{\pi}{2N}\right)^2$. Using this to simplify the above expression:

\begin{align}
    \mathbb E[\exp(\lambda_j^2/K^2)]&\leq \exp\left(\frac {l_{max}^2}{K^2}\right)+ f(0)\sum_{l\in R, |l|\leq l_{max} } \exp\left(-p\frac{l^2\pi^2}{8N^2}\right)\\
    &\leq  \exp\left(\frac {l_{max}^2}{K^2}\right)+ 2f(0)\sum_{l=l_{max}+1 }^{2n^2} \exp\left(-p\frac{l^2\pi^2}{8N^2}\right)\\
    &\leq  \exp\left(\frac {l_{max}^2}{K^2}\right)+ 2f(0)\exp\left(-p\frac{l_{max}^2\pi^2}{8N^2}\right)\sum_{l=0 }^{\infty} \exp\left(-p\frac{l_{max}\cdot l\pi^2}{8N^2}\right)\\
    &\leq  \exp\left(\frac {l_{max}^2}{K^2}\right)+ 2f(0)\exp\left(-p\frac{l_{max}^2\pi^2}{8N^2}\right)\cdot \frac 1{1-\exp\left(-p\frac{l_{max}\cdot \pi^2}{8N^2}\right)}\\
    &\leq  \exp\left(\frac {l_{max}^2}{K^2}\right)+ 2f(0)\exp\left(-p\frac{l_{max}^2\pi^2}{8N^2}\right)\cdot \frac 1{\left(\frac{pl_{max}\cdot \pi^2}{8N^2}\right)\left(1-\frac{pl_{max}\cdot \pi^2}{32n^4}\right)}&[\text{by \ref{lem_e-x_bound}}]
\end{align}

Now, let $l_{max}=\frac {4N}{\pi \sqrt p} l_0$ for some $l_0\geq 1$. Then, the above expression becomes upper bounded as 

\begin{align}
    \mathbb E[\exp(\lambda_j^2/K^2)]&\leq  \exp\left(\frac {l_0^2} {4K_0^2}  \right)+ 2f(0)\exp\left(-4l_0^2\right)\cdot \frac 1{\left(\frac{\sqrt pl_{0}\cdot \pi}{2N}\right)\left(1-\frac{\sqrt pl_{0}\cdot \pi}{4N}\right)}&[\text{by \ref{lem_e-x_bound}}]\\
    &\leq  \exp\left(\frac {l_0^2} {4K_0^2}  \right)+ \frac {4\sqrt 2 }{\pi}\cdot \exp\left(-4l_0^2\right)\cdot \frac {1}{l_0\left(1-\frac{\sqrt pl_{0}\cdot \pi}{4N}\right)}&[f(0)\leq\sqrt {2p}/N]\\
    &\leq \exp\left(\frac {l_0^2} {4K_0^2}  \right)+ \frac {4\sqrt 2 }{\pi}\exp\left(-4l_0^2\right)\cdot \frac 1{l_0\left(1-\frac{l_{0}\cdot \pi}{4}\right)}
\end{align}

Now, fix $l_0=1$. It can be verified with computation that the second term is at most $0.5$. With $K_0=1$, the first term is also at most $1.5$. Combining this, we get that Equation \ref{eqn_L_sbgnorm} holds with $K=\frac{4N}{\pi\sqrt p}$, so we have that $\norm{\lambda_j}_{\psi_2}\leq \frac {4n}{\sqrt p}$.
\begin{lemma}
    For any $\lambda_j$, the Sub-Gaussian norm is at most $C \cdot \frac{n}{\sqrt p}$ for some absolute constant $C$. Also, the maximum value of $|\lambda_j|$ is $N\cdot \frac \pi {2n}$
\end{lemma}

Here, the second part follows because because $|L|\leq N$.
\subsection{Analysis}

It is easy to see that we can now recreate the argument in Section \ref{sec_bnd_dv}, keeping in mind that now we consider 

\[\rho_{j+1}=e^{-\frac {\ii \lambda_j}{n}M_j}\rho_j e^{+\frac {\ii \lambda_j}{n}M_j}.\]

We again have that $\norm {\frac{\lambda_j}{n}M_j}\leq 1$, so the argument still works, up to a change in the absolute constants in Equation \ref{eqn_deviation}.

For the arguments in Section \ref{sec_cmb_all_tog}, we recreate them as follows.

We mainly need to show an analog of Theorem \ref{thm_sample_cmpl}. Instead of Lemma \ref{lem_single_estimation}, we will be using Lemma \ref{lem_low_mem_single}, for which we will work out the corresponding constraints.
\begin{theorem}\label{thm_smpl_cmpl_low_mem}
    The new Algorithm \ref{the_algorithm_low_mem}, for each $i\in [m]$ (individually), outputs an estimate for $\mathbb E_{\rho}[M_i]$ which has an additive error of at most $\epsilon$, with probability at least $1-\delta$, when the following constraints hold (with appropriate absolute constants):
    \begin{enumerate}
        \item $n\geq C_0\frac {\log(1/\delta)}{\epsilon^2}$ (and $n$ is a power of $2$).
        \item $p\leq C_1\frac{n^2\epsilon^2}{\log(1/\delta)}$
        \item  $p\geq C_3\frac {m\log (1/\delta)}{\epsilon^2}$.
        \item $p\geq C_4\frac{\log (1/\delta)}{\epsilon} +C_5 \frac {m}{\epsilon} $
    \end{enumerate}
\end{theorem}
\begin{proofof}{Theorem \ref{thm_smpl_cmpl_low_mem}}
    From the first two constraints, we get that Lemma \ref{lem_low_mem_single} can invoked. So, at the $i\ith$ step, we are able to get an estimate of $\mathbb E_{\rho_i}[M_i]$ with additive error at most $\epsilon$ with probability at least $1-\delta$. So, it only remains to bound the deviation of $\mathbb E_{\rho_i}[M_i]$ from $\mathbb E_{\rho}[M_i]$.

Notice that in the proof of Theorem \ref{thm_sample_cmpl}, all the equations from \ref{eqn_telescoped_bgn} to \ref{eqn_telescoped_difference} follow up to changes in the constants. 

We rewrite Equation \ref{eqn_telescoped_difference} for simplicity:

\[ |\Tr[M_i\rho_i]-\Tr[M_i\rho]|\leq  \frac {c_1}{n}\left|\sum_{j=1}^{i-1}\lambda_j \Tr[i[M_i,M_j]\rho_j]\right| + \frac {c_2}{n^2}\sum_{j=1}^{i-1}\lambda_j^2 .\]

As in the proof for $\ref{thm_sample_cmpl}$, we again first bound the first term.  Define again $q_j\eqdef \Tr[i[M_i,M_j]\rho_j]$, which is always at most 2. Notice that $q_j$ only depends on $\lambda_1,\lambda_2\dots \lambda_{j-1}$. We denote this set of random variables as $\lambda_{1:j-1}$. With this, we try to bound the first term using an M.G.F. argument. Let $K\eqdef \norm{\lambda_j}_{\psi_2}$. Choose $\alpha>0$. In this argument, we will use $\Tilde C$ to denote absolute constants (multiple uses can correspond to possibly different absolute constants). We will use properties from \cite{HDP} about Sub-Gaussian distributions, specifically the ones in Proposition 2.5.2. We have

\begin{align}
    \Pr\left[\sum_{j=1}^{i-1}\lambda_jq_j\leq \frac {n\epsilon}{c_1} \right]&\leq \Pr\left[\exp\left(\alpha \sum_{j=1}^{i-1}\lambda_jq_j\right)\geq e^{\alpha \frac {n\epsilon }{c_1}}\right]\\
    &\leq  \mathbb E[\exp\left( \alpha \sum_{j=1}^{i-1}\lambda_jq_j\right)]\cdot e^{-\alpha\frac {n\epsilon }{c_1}}\\
     &\leq \mathbb E[\exp\left( \alpha\sum_{j=1}^{i-2}\lambda_jq_j\right)\cdot \mathbb E[\exp\left(\alpha \lambda_jq_j\right)|\lambda_{1:j-1}]]\cdot e^{-\alpha \frac {n\epsilon }{c_1}}\\
     &\leq \mathbb E[\exp\left( \alpha\sum_{j=1}^{i-2}\lambda_jq_j\right)\cdot \exp(\Tilde C\alpha^2 K^2) \cdot e^{-\alpha \frac {n\epsilon }{c_1}}\\
     &\quad \quad \quad \quad \quad [\text{using property (v) in Proposition 2.5.2 of \cite{HDP}}]   \nonumber\\
     &\leq  \exp(\Tilde C(i-1)\alpha^2 K^2) \cdot e^{-\alpha \frac {n\epsilon }{c_1}}\\
     &\quad \quad \quad \quad \quad [\text{repeating $i-1$ times}]\nonumber
\end{align}
This holds for all $\alpha>0$, so we can minimize the last expression over $\alpha$ to get 

\begin{equation}
    \Pr\left[\sum_{j=1}^{i-1}\lambda_jq_j\leq \frac {n\epsilon}{c_1} \right]\leq \exp\left(-\Tilde C\frac {n^2\epsilon^2}{(i-1)K^2} \right)\leq \exp\left(-\Tilde C \frac {p\epsilon^2}{m}\right) 
\end{equation}

Using the same argument for the other side, we get

\begin{equation}
    \Pr\left[\frac {c_1}n\left|\sum_{j=1}^{i-1}\lambda_jq_j\right|\leq \epsilon\right]\leq 2\exp\left(-\Tilde C \frac {p\epsilon^2}{m}\right) 
\end{equation}

By the given constraints, this is at most $2\delta$.

For the second term in \ref{eqn_telescoped_difference}, we again see that 

\[\mathbb E\left[\frac {\sum_{j=1}^{i-1} \lambda_j^2}{K^2}\right]=2^{i-1}\]

So we have that 

\begin{align*}
    \Pr\left[\frac {c_2}{n^2} \sum_{j=1}^{i-1} \lambda_j^2\leq \epsilon \right] &\leq 2^{i-1}\cdot \exp\left(-\Tilde C\frac {n^2\epsilon } {K^2}  \right) \\
    &\leq2^{m}\cdot \exp\left(-\Tilde C{p\epsilon } \right) 
\end{align*}
By the given constraints, this is at most $\delta$.

So, overall, at the $i\ith$ stage, we receive an estimate with additive error at most $3\epsilon$ and probability of failure at most $4\delta$. Scaling $\epsilon$ and $\delta$ appropriately in the argument (with appropriate changes in the absolute constants) gives us the result.

\end{proofof}

By an argument similar to Corollary \ref{cor_main}, we see that the sample complexity of Algorithm \ref{the_algorithm_low_mem} is also the same.

\subsection{Implementation details}\label{sec_low_mem_implementation}

We see how the properties discussed in Section \ref{sec_other_props} change:

\begin{enumerate}
	\item The classical post-processing for Algorithm \ref{the_algorithm_low_mem} is still trivial.
	\item Unlike Algorithm \ref{the_algorithm}, Algorithm \ref{the_algorithm_low_mem} is no longer robust against qubit measurement noise. This is because, amongst the $2\log n +2$ qubits, an error in the most significant bit can swing the value of the output estimate dramatically.
	\item For the gate complexity of the algorithm, the $\Theta(n\cdot \sum S_i)$ term is unchanged, as we still need to store the answer of each measurement in an ancilla qubit. However, the unitary operator in Equation \ref{unitary_disc} can now be split up into products of $\exp(\ii M^{(i)}\otimes T)$, each of which can be implemented as a controlled 1-addition, controlled by the ancilla qubit containing the answer of $M$ on the $i\ith$ copy of $\rho$, and the target being the register $B$. This has complexity $O(\log n)$ for each copy of $\rho$, so overall complexity for this is $O(m\cdot n\cdot \log n)$, making the overall gate complexity $O(n\cdot m\cdot \log n +n\cdot\sum S_i)$, apart from the cost of constructing the state $\ket{\psi_p}$ $m$ times, which will now have cost $O(m\cdot \mathrm{polylog}(n))$.
	
	\item The implementation as a read once circuit can now be done with $O(m\log n)$ working memory, as we replace operations on the $k$-ancilla qubits with the operations on the $B$ register with $O(\log n)$ qubits.
	
	\item Implementation with $O(1)$ additional memory is no longer directly possible for Algorithm \ref{the_algorithm_low_mem} as each estimation required the register $B$ with $O(\log n)$ bits.
	\item Implementation with $O(\log n )$ memory is achieved more directly by simply reusing the $B$ register for all rounds.
\end{enumerate}

\end{document}